\documentclass[letterpaper, 10 pt, conference]{ieeeconf}  %

\IEEEoverridecommandlockouts                              %

\usepackage{graphics} %
\usepackage{epsfig} %
\usepackage{mathptmx} %
\usepackage{times} %
\usepackage{amsmath} %
\usepackage{amssymb}  %

\usepackage[ansinew]{inputenc} %
\usepackage[T1]{fontenc} %
\usepackage{xcolor}
\usepackage{MnSymbol}
\usepackage{mathtools}
\usepackage{bm}
\usepackage{booktabs}
\usepackage{microtype} %
\usepackage{siunitx}
\usepackage{tabularx, tabulary}
\usepackage{ragged2e}
\usepackage{multirow}
\usepackage[linesnumbered,ruled,vlined]{algorithm2e}
\let\labelindent\relax %
\usepackage{enumitem}
\usepackage{hyperref}
\usepackage{cite}
\usepackage{standalone}

\usepackage{tikz}
\newcommand\tikznode[3][]%
{\tikz[remember picture,baseline=(#2.base)]
	\node[minimum size=0pt,inner sep=0pt,#1](#2){#3};%
}
\usepackage{calc}
\usetikzlibrary{calc} 

\usetikzlibrary{decorations} %
\usetikzlibrary{matrix} %
\usetikzlibrary{arrows} %
\usetikzlibrary{backgrounds, fit}
\usepackage{tikz-layers}

\usetikzlibrary{positioning,shapes.multipart,shapes.arrows, decorations.markings,arrows.meta,shapes.symbols}

\tikzset{
	-|-/.style={
		to path={
			(\tikztostart) -| ($(\tikztostart)!#1!(\tikztotarget)$) |- (\tikztotarget)
			\tikztonodes
		}
	},
	-|-/.default=0.5,
	|-|/.style={
		to path={
			(\tikztostart) |- ($(\tikztostart)!#1!(\tikztotarget)$) -| (\tikztotarget)
			\tikztonodes
		}
	},
	|-|/.default=0.5,
}

\tikzstyle{block} = [draw=black,fill=white,rectangle,thick,minimum height=2em,minimum width=4em, align=center]
\tikzstyle{sum} = [draw,circle,inner sep=0mm,minimum size=2mm]
\tikzstyle{connector} = [->,thick]
\tikzstyle{dashedconnector} = [->,thick,dashed]
\tikzstyle{line} = [thick]
\tikzstyle{dashedline} = [thick,dashed]
\tikzstyle{branch} = [circle,inner sep=0pt,minimum size=1mm,fill=black,draw=black]
\tikzstyle{guide} = []
\tikzstyle{snakeline} = [connector, decorate, decoration={pre length=0.2cm,
	post length=0.2cm, snake, amplitude=.4mm,
	segment length=2mm},thick, magenta, ->]

\newcommand{\gettikzxy}[3]{%
	\tikz@scan@one@point\pgfutil@firstofone#1\relax
	\edef#2{\the\pgf@x}%
	\edef#3{\the\pgf@y}%
}

\makeatletter
\pgfkeys{/pgf/.cd,
	parallelepiped offset x/.initial=2mm,
	parallelepiped offset y/.initial=2mm
}
\pgfdeclareshape{parallelepiped}
{
	\inheritsavedanchors[from=rectangle] %
	\inheritanchorborder[from=rectangle]
	\inheritanchor[from=rectangle]{north}
	\inheritanchor[from=rectangle]{north west}
	\inheritanchor[from=rectangle]{north east}
	\inheritanchor[from=rectangle]{center}
	\inheritanchor[from=rectangle]{west}
	\inheritanchor[from=rectangle]{east}
	\inheritanchor[from=rectangle]{mid}
	\inheritanchor[from=rectangle]{mid west}
	\inheritanchor[from=rectangle]{mid east}
	\inheritanchor[from=rectangle]{base}
	\inheritanchor[from=rectangle]{base west}
	\inheritanchor[from=rectangle]{base east}
	\inheritanchor[from=rectangle]{south}
	\inheritanchor[from=rectangle]{south west}
	\inheritanchor[from=rectangle]{south east}
	\backgroundpath{
		\southwest \pgf@xa=\pgf@x \pgf@ya=\pgf@y
		\northeast \pgf@xb=\pgf@x \pgf@yb=\pgf@y
		\pgfmathsetlength\pgfutil@tempdima{\pgfkeysvalueof{/pgf/parallelepiped
				offset x}}
		\pgfmathsetlength\pgfutil@tempdimb{\pgfkeysvalueof{/pgf/parallelepiped
				offset y}}
		\def\ppd@offset{\pgfpoint{\pgfutil@tempdima}{\pgfutil@tempdimb}}
		\pgfpathmoveto{\pgfqpoint{\pgf@xa}{\pgf@ya}}
		\pgfpathlineto{\pgfqpoint{\pgf@xb}{\pgf@ya}}
		\pgfpathlineto{\pgfqpoint{\pgf@xb}{\pgf@yb}}
		\pgfpathlineto{\pgfqpoint{\pgf@xa}{\pgf@yb}}
		\pgfpathclose
		\pgfpathmoveto{\pgfqpoint{\pgf@xb}{\pgf@ya}}
		\pgfpathlineto{\pgfpointadd{\pgfpoint{\pgf@xb}{\pgf@ya}}{\ppd@offset}}
		\pgfpathlineto{\pgfpointadd{\pgfpoint{\pgf@xb}{\pgf@yb}}{\ppd@offset}}
		\pgfpathlineto{\pgfpointadd{\pgfpoint{\pgf@xa}{\pgf@yb}}{\ppd@offset}}
		\pgfpathlineto{\pgfqpoint{\pgf@xa}{\pgf@yb}}
		\pgfpathmoveto{\pgfqpoint{\pgf@xb}{\pgf@yb}}
		\pgfpathlineto{\pgfpointadd{\pgfpoint{\pgf@xb}{\pgf@yb}}{\ppd@offset}}
	}
}
\makeatother

\tikzstyle{server}=[
parallelepiped,
fill=white, draw=white,
minimum width=0.35cm,
minimum height=0.75cm,
parallelepiped offset x=3mm,
parallelepiped offset y=2mm,
xscale=-1,
path picture={
	\draw[top color=white,bottom color=white]
	(path picture bounding box.south west) rectangle 
	(path picture bounding box.north east);
	\coordinate (A-center) at ($(path picture bounding box.center)!0!(path
	picture bounding box.south)$);
	\coordinate (A-west) at ([xshift=-0.575cm]path picture bounding box.west);
	\draw[ports]([yshift=0.1cm]$(A-west)!0!(A-center)$)
	rectangle +(0.2,0.065);
	\draw[ports]([yshift=0.01cm]$(A-west)!0.085!(A-center)$)
	rectangle +(0.15,0.05);
	\fill[white]([yshift=-0.35cm]$(A-west)!-0.1!(A-center)$)
	rectangle +(0.235,0.0175);
	\fill[white]([yshift=-0.385cm]$(A-west)!-0.1!(A-center)$)
	rectangle +(0.235,0.0175);
	\fill[white]([yshift=-0.42cm]$(A-west)!-0.1!(A-center)$)
	rectangle +(0.235,0.0175);
}
]

\tikzstyle{ports}=[
line width=0.3pt,
top color=white,
bottom color=white
]

\usepackage{pgfplots}
\pgfplotsset{compat=newest}
\pgfplotsset{plot coordinates/math parser=false}
\newlength\figureheight
\newlength\figurewidth 

\usepackage{amsmath,amssymb,amsfonts}    %
\usepackage{mathtools}
\newcommand{\R}{\mathbb{R}} 

\newcommand{\N}{\mathbb{N}}

\newcommand{\E}{\mathbb{E}}

\newcommand{\Abc}{\boldsymbol{\mathcal{A}}}
\newcommand{\AbcTilde}{\boldsymbol{\mathcal{\tilde{A}}}}
\newcommand{\Bbc}{\boldsymbol{\mathcal{B}}}
\newcommand{\Cbc}{\boldsymbol{\mathcal{C}}}
\newcommand{\Dbc}{\boldsymbol{\mathcal{D}}}

\newcommand{\Ic}{\mathcal{I}}

\newcommand{\modq}{\bmod{q}}
\renewcommand{\mod}[1]{\,(\mathrm{mod}\,#1)}

\newcommand{\ct}{\mathsf{ct}}
\newcommand{\sk}{\mathsf{sk}}

\newcommand{\Enc}{\mathsf{Enc}}
\newcommand{\pk}{\mathsf{pk}}

\newcommand{\Dec}{\mathsf{Dec}}

\newcommand{\RS}{\mathsf{RS}}

\newcommand{\Mult}{\mathsf{Mult}}

\newcommand{\round}[1]{\left\lfloor#1\right\rceil}
\newcommand{\roundtext}[1]{\lfloor#1\rceil}

\newcommand{\innerproduct}[2]{\left\langle #1, #2 \right\rangle}

\definecolor{istblue}{rgb/cmyk}{0.0,0.25490,0.56862745/1,0.55,0.0,0.43}%
\definecolor{istgreen}{rgb/cmyk}{0.388235,0.83137,0.4431372/0.53,0,0.46,0.16}%
\definecolor{istorange}{rgb/cmyk}{1.0,0.5333,0.0667/0.0,0.46,0.93,0.0}%
\definecolor{istred}{rgb/cmyk}{0.9961,0.2902,0.2863/0.0,0.7,0.71,0.0}%
\definecolor{istlightblue}{rgb/cmyk}{0.3765, 0.6863, 1.0/0.62,0.31,0.0}%
\definecolor{istdarkblue}{rgb/cmyk}{0.1176,0.1804,0.8706/0.86,0.79,0.0,0.12}%
\definecolor{istdarkgreen}{rgb/cmyk}{0.0627,0.5882,0.2824/0.89,0.0,0.52,0.41}%
\definecolor{istdarkred}{rgb/cmyk}{0.529,0.031,0.075/0,0.94,0.86,0.47}%
\definecolor{istlogoblue}{rgb/cmyk/gray}{0,0,0.804/1,1,0,0.2/0}%
\definecolor{unianthrazit}{rgb/cmyk}{0.24314,0.26667,0.29804/0.5,0.35,0.25,0.70}%
\definecolor{unimiddleblue}{rgb/cmyk}{0,0.31765,0.61961/1,0.7,0,0}%
\definecolor{unilightblue}{rgb/cmyk}{0,0.74510,1/0.7,0,0,0}%

\newcommand{\diagdots}[3][-25]{%
	\rotatebox{#1}{\makebox[0pt]{\makebox[#2]{\xleaders\hbox{$\cdot$\hskip#3}\hfill\kern0pt}}}%
}

\newtheorem{assumption}{Assumption}
\newtheorem{definition}{Definition}
\newtheorem{theorem}{Theorem}
\newtheorem{lemma}{Lemma}

\newtheorem{problem}{Problem}

\title{\LARGE \bf
Bootstrapping Guarantees: Stability and Performance Analysis\\
for Dynamic Encrypted Control
}

\author{Sebastian Schlor and Frank Allgöwer%
\thanks{F.\ Allgöwer is thankful that this work was funded by the Deutsche Forschungsgemeinschaft (DFG, German Research Foundation) under Germany's Excellence Strategy -- EXC 2075 -- 390740016 and within grant AL 316/13-2 -- 285825138 and AL 316/15-1 -- 468094890. 
S.\ Schlor thanks the Graduate Academy of the SC SimTech for its support.}%
\thanks{ S.\ Schlor and F.\ Allgöwer are with the University of Stuttgart, Institute for Systems Theory and Automatic Control, Germany.  {\tt\small \{schlor,allgower\}@ist.uni-stuttgart.de}.}%
}

\begin{document}

\maketitle
\thispagestyle{empty}
\pagestyle{empty}

\begin{abstract}

Encrypted dynamic controllers that operate for an unlimited time have been a challenging subject of research.
The fundamental difficulty is the accumulation of errors and scaling factors in the internal state during operation.
Bootstrapping, a technique commonly employed in fully homomorphic cryptosystems, can be used to avoid overflows in the controller state but can potentially introduce significant numerical errors. 
In this paper, we analyze dynamic encrypted control with explicit consideration of bootstrapping.
By recognizing the bootstrapping errors occurring in the controller's state as an uncertainty in the robust control framework, we can provide stability and performance guarantees for the whole encrypted control system.
Further, the conservatism of the stability and performance test is reduced by using a lifted version of the control system.

\end{abstract}

\section{INTRODUCTION}

Encrypted control offers the ability to outsource the computations to evaluate the control law to external servers while maintaining the privacy of the involved data at the same time. This not only enhances flexibility and reduces the need for dedicated hardware for the controller, but also can be valuable for control and monitoring of distributed systems.
From a cryptographic point of view, the technology that enables end-to-end encrypted computations is homomorphic encryption. 
A simple example for this concept is, e.g., the Paillier cryptosystem~\cite{Paillier1999}, which is \emph{additively homomorphic}, i.e., there exists an operation $\oplus$ such that
\begin{align}
	\Dec(\Enc(x_1)\oplus\Enc(x_2)) = x_1+x_2.
\end{align}
Here, $\Enc$ and $\Dec$ correspond to the encryption and decryption process, respectively.
\emph{Leveled homomorphic} arithmetic cryptosystems allow for the evaluation of any fixed-degree polynomial. 
They become \emph{fully homomorphic} if they extend to arbitrary many operations.

Modern homomorphic cryptosystems for arithmetic operations on real numbers are approximate, i.e., the result of the computation is not exact but contains a small error. This enhances security and can be mitigated by appropriate scaling and rounding of the decrypted result.
Furthermore, cryptosystems typically only support integers as inputs, which means that every real-number input has to be quantized and represented as a fixed-point number with integer representation and an corresponding scaling factor.

While this works well for predetermined computations, encrypted control raises additional challenges.
Depending on the application, the results should be available in real-time.
Further, in dynamic control, recursive computations are performed over a possibly infinite time-horizon, which exceeds preset degrees for the supported polynomials. 
This blows up the introduced errors so that the decryption of the result is no longer valid. Ever-growing scaling factors, especially through multiplications, lead to undetected overflows in the results.

\subsection{Related work}

To mitigate these effects, researchers from the control community have come up with different measures. 
When the recursive computations are about to exceed the predefined number of multiplications, sending the encrypted controller state to the plant and receiving a freshly rounded and re-encrypted state from the plant, as used in~\cite{Stobbe2022}, can reset the error and the scaling factor. However, this adds more computational burden to the plant, and additional communication is needed.
In contrast,~\cite{Murguia2020} proposed a periodic reset of the controller state to a known encryption of zero. This solves the problem of overflow and adds no further communication and computation; however, a transient phase of convergence is introduced repeatedly, which can lead to unsatisfactory control behavior.
A similar line of thought lead to the proposal of~\cite{Schluter2021}, where instead of periodic resets of the whole state, an FIR filter-type controller was used. There, the output of the FIR filter always contains only a constant number of factors, which can be computed with a predefined number of multiplications.
To avoid scaling factors entirely,~\cite{Kim2021,Kim2022,Lee2023a} proposed to transform non-integer system matrices into integer matrices using pole-placement by feedback with the plant. The integer reformulation, however, usually leads to 
the cloud-implemented controller being unstable~\cite{Schluter2022}. Merely the feedback stabilizes the controller dynamics. Additionally, re-encryption and communication by the actuator is necessary.

From a cryptographer's point of view, the first solution to arbitrary many computations on encrypted data was proposed by~\cite{gentry2009fully} introducing so-called bootstrapping. The main idea of this procedure is to evaluate the decryption process, which subsequently enables refreshing the ciphertext, in an encrypted fashion. For an overview of recent progress, we refer to~\cite{Badawi2023,Marcolla2022a}.
In more modern fully homomorphic cryptosystems, such as CKKS\cite{Cheon2017}, which we build on in this paper, the core of bootstrapping corresponds to a so-called modular reduction, which entails a polynomial approximation of the modulo function. This is one of the sources of numerical accuracy loss. We give a quick introduction in Section~\ref{subs:Bootstrapping}. 
Recent research on bootstrapping has focused on different advanced polynomial approximations of the modulo function, mainly with respect to minimizing the point-wise maximal error~\cite{Jutla2020,Lee2020} or its $\ell_2$-norm~\cite{Lee2022}. 
While this works well for general purpose applications, only a relative error measure can ensure asymptotic stability for encrypted control.
In~\cite{Kim2016}, bootstrapping was used for encrypted control, but only timing aspects were considered.

\subsection{Contribution}
This is the first work to incorporate bootstrapping in the analysis of encrypted control. In particular, we make the following contributions:
\begin{itemize}
	\item For encrypted dynamic control, we explicitly take the bootstrapping errors into account.
	\item By considering these errors as static sector nonlinearities in the control loop, we provide a stability and performance test for the encrypted control system using the robust control framework.
	\item By lifting the system dynamics,
	we further reduce the conservatism of the stability and performance guarantees.
	\item With our framework, we unify the analysis of encrypted control approaches with bootstrapping, periodic state reset, and FIR filters.
\end{itemize}
Our analysis shows that encrypted control has different requirements on the bootstrapping than general purpose homomorphic encryption. Using tailored bootstrapping polynomails for encrypted control can reduce the computation time and enhance the control performance.

\section{Preliminaries}

In this section, we present the most relevant properties of cryptosystems for the focus of the paper. 

\subsection{Notation}
For $q\in\N$, we define the centered modulo operator by $m \modq=m-q\roundtext{m/q}$, with the rounding operator $\roundtext{\cdot}$.
By $\innerproduct{a}{s}$ we denote the inner product of the vectors $a$ and $s$.
By $0$ and $I$ we we denote the zero and identity matrix of matching dimensions.
By $\succ 0$ ($\succeq 0$) we denote positive {(semi\nobreakdash-)} definiteness.
We abbreviate terms by $(\star)$ if they can be obtained by symmetry.

\subsection{CKKS}

One of the most advanced fully homomorphic cryptosystems supporting arithmetic operations on approximate real numbers is the Cheon-Kim-Kim-Song (CKKS) scheme~\cite{Cheon2017}, which is a special type of Learning With Errors (LWE) cryptosystem~\cite{Regev2009}.
Here, we only introduce a simplified, abstract version thereof to highlight the most important functionalities for encrypted control.

Before encrypting and manipulating encrypted numbers, every number or vector of numbers has to be encoded as a polynomial in a polynomial ring. 
For the sake of clarity, we  will not go into more detail here. For this paper, we treat them as if they were numbers or vectors of reals or integers, i.e., we can add, multiply, and round them and they have an inner product.

For decryption and encryption, a secret and a public key are created. The secret key is denoted by $\sk = (1,s)$, where $s$ is $n$-dimensional and sampled from a predefined distribution. For a public modulus $q$, the public key is generated as $\pk = (b,a)$ with $b=-\innerproduct{a}{s}+e\bmod{q}$, $a$ of equal dimension and random, and $e$ a random scalar.
The encryption of a secret number $m$ is then generated as $\Enc_{\pk}(m) = (m,0) + \pk \bmod{q}= (-\innerproduct{a}{s}+m+e\bmod{q}, a) \eqcolon \ct_{q}(m)$. 
Decryption works by taking the inner product $\Dec_\sk(\ct_{q}(m))=\innerproduct{\ct_{q}(m)}{\sk}\bmod{q} = -\innerproduct{a}{s}+m+e + \innerproduct{a}{s} \bmod{q}= m+e\bmod{q}\approx m$ if $e$ is small and $m\in(-\tfrac{q}{2}, \tfrac{q}{2}]$.
Further, CKKS has operations for addition and multiplication with public and secret factors implemented. For more details, we refer to the literature~\cite{Marcolla2022a}.

\subsection{Integer representation and rescaling}\label{subs:Rescaling}

Every number that is encrypted in CKKS has to be represented by an integer first, before the encoding and encryption.
This can be done by appropriate scaling and rounding.
Suppose we have a number $m_\mathrm{real}\in \R$. We can obtain an approximate scaled and rounded integer representation as $m = \round{c m_\mathrm{real}}$ with some scaling factor $c$, which has to be remembered.
In the setup, a chain of moduli of the ciphertext $q_\ell = q_0c^\ell \in\{q_0,\dots, q, \dots, Q, \dots, q_L\}$ with levels $\ell$ is selected. 
Now suppose we have two numbers $m_1$ and $m_2$ 
represented in this form with scaling factors $c$. Then, the product $\Mult(\ct_{q_{\ell}}(m_1),\ct_{q_{\ell}}(m_2)) = \ct_{q_{\ell}}(m_3)$ 
contains a scaling factor $c^2$ to be represented correctly.
Thus, coined to dynamic encrypted control, if the controller's system matrix contains non-integer values, the scaling factor of the controller state is ever-growing. This leads to overflows.
To reduce the digits and scaling factor again, a rescaling operation $\RS(\ct_{q_{\ell}}(m_3)) = \ct_{q_{\ell-1}}(m_3) = \round{\frac{1}{c}\ct_{q_{\ell}}(m_3)}\bmod{q_{\ell-1}}$ is implemented. This reduces the scaling factor back to $c$ at the cost of reducing the modulus. Since numbers cannot be represented correctly anymore if the modulus gets too small, no more rescaling is possible at some point.

\subsection{Bootstrapping}\label{subs:Bootstrapping}
The bootstrapping operation is needed to raise the modulus again and to enable further computations.
For that, we first interpret the ciphertext $\ct_{q}(m)$ with small modulus $q$ as if it was given in a larger modulus $Q$.
The decrypted value would result in $\Dec_\sk(\ct_{q \text{~(but assume $Q$)}}(m)) = (-\innerproduct{a}{s}+m+e \modq) + \innerproduct{a}{s} = (-\innerproduct{a}{s}+m+e -rq) + \innerproduct{a}{s} = m+e -rq \bmod{Q} \approx m-rq$ with the number of overflows in the small modulus $r = \roundtext{\tfrac{-\innerproduct{a}{s}+m+e}{q}}$.
However, this number $r$ is unknown due to the encryption.
To remove the rescaling error $rq$ and get a correct representation of $\ct_{Q}(m)$, the essence of bootstrapping is evaluating $m\modq \mod{Q}$ homomorphically.

However, since the modulo function is non-polynomial, and the homomorphisms of the cryptosystem only allow addition and multiplication, a polynomial approximation has to be evaluated.
In the literature, scaled versions of Chebyshev polynomials~\cite{Cheon2018a} and Taylor and Chebyshev approximations~\cite{Chen2019}, as well as polynomials minimizing the point-wise maximal error~\cite{Jutla2020,Lee2020} or its $\ell_2$ norm~\cite{Lee2022} were proposed.
Here, we take polynomials with a relative error description and show their benefits for encrypted control.

\section{Bootstrapping polynomial}\label{sec:BootError}

For the analysis of the errors introduced by bootstrapping with an approximating polynomial we make the following assumptions.
\begin{assumption}~\\[-1.1\baselineskip]
	\begin{enumerate}
		\item The encrypted values $|m+e|\leq\epsilon\frac{q}{2}<\frac{q}{2}$ are bounded by $\epsilon\in(0,1)$ relative to the maximal representable number. 
		\item The number of overflows $|r|\leq K$ is bounded by $K\in\N$. 
	\end{enumerate}
\end{assumption}

To satisfy the first assumption, the modulus $q$ is set up large enough, and the controller design has to ensure stability or, more precisely, invariance.
As the wrap-arounds are caused by the inner product $\innerproduct{a}{s}$ with bounded and random $a$ and $s$, $r$ is drawn from an Irwin-Hall distribution.
Due to its finite support, the second assumption is also satisfied.

From this observation, we see that the modulo function only has to be accurately approximated in intervals of width $\epsilon\frac{q}{2}$ around multiples of $q$ up to $\pm Kq$, i.e., on the set $\Ic=\{m-rq\,:\,|m|\leq \epsilon\frac{q}{2}, r\in\{-K,\ldots,K\}\}$.

Let $p:\R\to\R, p(m)= \sum_{i=0}^d \tilde{p}_i m^i$ of fixed degree $d$ and with coefficient vector $\tilde{p}$ be the polynomial that approximates the modulo function.
Then, the error between the polynomial and the modulo function can be measured in different ways.
For general purpose homomorphic encryption, $|p(m)-(m\modq)|\leq\gamma ~\forall m\in\Ic$ bounds the maximal point-wise error by $\gamma$ (c.f.\ \cite{Jutla2020,Lee2020}). This is useful to give absolute worst-case error bounds.
The expected variance of the error can be described by $\E_\Ic(|p(m)-(m\modq)|^2)$ (c.f.\ \cite{Lee2022}).

In contrast to the literature on general homomorphic encryption, we use a relative error measure $|p(m)-(m\modq)|\leq\gamma\,|m\modq| ~\forall m\in\Ic$, which is necessary to ensure stability of the encrypted control system.
A polynomial satisfying this kind of error bound can be generated by solving an interpolation problem, e.g., Hermite interpolation with conditions on the value of $p(m)$ and its derivative $p^\prime(m)$ at $m = rq, r\in\{-K,\dots,K\}$, or interpolation over sampled points.
A more rigorous alternative would be to use an optimization-based approach by sum-of-squares optimization. To this end, the optimal parameters can be obtained by
\begin{align}\label{eq:minimaxPoly}
	\tilde{p}^\star  =  \arg\min_{\tilde{p},\gamma} & \quad\hspace*{-0.5em} \gamma\\
	\text{s.t.} & \quad\hspace*{-0.5em} -\gamma m \leq m-p(m-rq) \leq \gamma m  &\hspace*{-0.5em}\forall m\in [-\epsilon\frac{q}{2},\epsilon\frac{q}{2}],\notag\\
	 &&\hspace*{-0.5em}{}\phantom{\forall }r\in\{-K,\ldots,K\}. \notag
\end{align}
Here, we demonstrate the concept by a optimized polynomial with degree $d=25$, the upper bound $K=2$, and the number range $\epsilon = 0.5$, which leads to a polynomial that satisfies the sector bound with $\gamma = 0.2296$. The polynomial is depicted next to the modulo function in Fig.~\ref{fig:modPoly}. In Fig.~\ref{fig:sector}, the error between the modulo function and the polynomial is shown over $m\modq$ instead of $m$.

There is a natural trade-off between precision, i.e., a small $\gamma$, and the evaluation complexity in terms of the degree $d$ of the polynomial. 
Thus, if the control system can cope with larger bootstrapping errors, a smaller-degree polynomial with quicker evaluation time can be used, which helps in real-time applications.

\begin{figure}[t]
	\centering
\setlength\figurewidth{0.8\columnwidth}
\setlength\figureheight{0.5\columnwidth}
\input{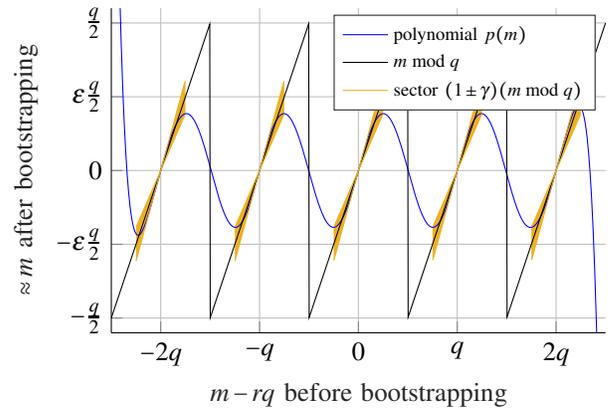} 
	\caption{The modulo function and its polynomial approximation for bootstrapping.}\label{fig:modPoly}
\end{figure}
\begin{figure}[t]
\centering
\setlength\figurewidth{0.8\columnwidth}
\setlength\figureheight{0.5\columnwidth}
\input{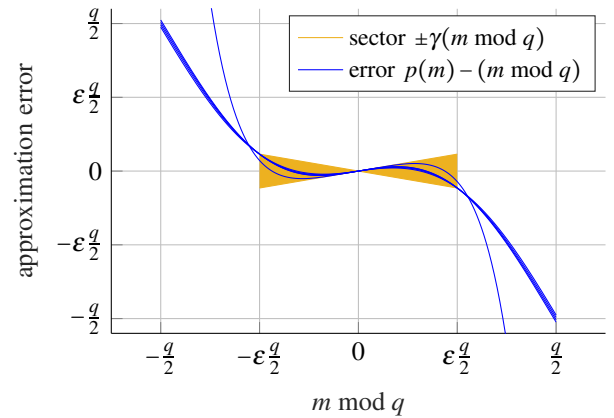} 
	\caption{Relative error of the polynomial approximation to the modulo function for bootstrapping.
	The figure shows multiple error functions since the bootstrapping polynomial in Fig.~\ref{fig:modPoly} is evaluated at different intervals depending on the offset $rq$.
	}\label{fig:sector}
\end{figure}

\section{Problem description and formulation as robust control problem}\label{sec:setup}

\begin{figure}[t]
	
	\centering
	\begin{tikzpicture}[scale=1, auto, >=stealth']
		\def\blockHeight{2\baselineskip}
		\node[block, minimum height=\blockHeight]  at (0,0)  (plant) {\parbox{8em}{\hfill $x$\\}};
		\node[]  at (plant)  (plantText) {\centering Plant};
		\node[block, minimum height=\blockHeight, below = 1\baselineskip of plant] (controller) {\parbox{8em}{\hfill $x_c$\\}};
		\node[]  at (controller)  (controllerText) {\centering Controller};
		\node[block, minimum height=\blockHeight, below = 1\baselineskip of controller] (sector) {\parbox{8em}{\hfill
			 }};
		\node[]  at (sector)  (sectorText) {\centering Bootstrapping};
		
		\node[guide, left=0.5cm of plant] (guideLeft) {};
		\node[guide, right=0.5cm of plant] (guideRight) {};
		\node[guide, left=2cm of plant] (guideLeftLeft) {};
		\node[guide, right=2cm of plant] (guideRightRight) {};
		
		\draw[connector]  ($(guideLeftLeft)+(0,+0.6\baselineskip)$) -- node[pos=0.1]{$w_{p_1}$} ($(plant.west)+(0,+0.6\baselineskip)$);
		\draw[connector]  ($(plant.east)+(0,+0.6\baselineskip)$) -- node[pos=0.9]{$z_p$}  ($(guideRightRight)+(0,+0.6\baselineskip)$);
		
		\draw[connector]  ($(controller.west)+(0,+0.6\baselineskip)$) -| node[pos=0.75,left]{$u$} ($(plant.west)+(-0.5cm,-0.6\baselineskip)$) --  ($(plant.west)+(0,-0.6\baselineskip)$);
		\draw[connector]  ($(plant.east)+(0,-0.6\baselineskip)$) --  ($(plant.east)+(0.5cm,-0.6\baselineskip)$) |-  node[pos=0.25,right]{$y$} ($(controller.east)+(0,+0.6\baselineskip)$);
		
		\draw[connector]  (sector.east) -| node[pos=0.75,right]{$w_u$} ($(controller.east)+(0.5cm,-0.6\baselineskip)$) -- ($(controller.east)+(0,-0.6\baselineskip)$);
		\draw[connector]  ($(controller.west)+(0,-0.6\baselineskip)$) --  ($(controller.west)+(-0.5cm,-0.6\baselineskip)$) |- node[pos=0.25,left]{$z_u$} (sector.west);
		
		\draw[connector]  ($(controller.east)+(2cm,0)$) 
		-- 
		node[pos=0.1,above]{$w_{p_2}$} (controller.east);
		
	\end{tikzpicture}
	\caption{Block diagram of the encrypted control system with bootstrapping. The bootstrapping is is interpreted as static nonlinearity acting on the controller state.}\label{fig:blocks}
	
\end{figure}
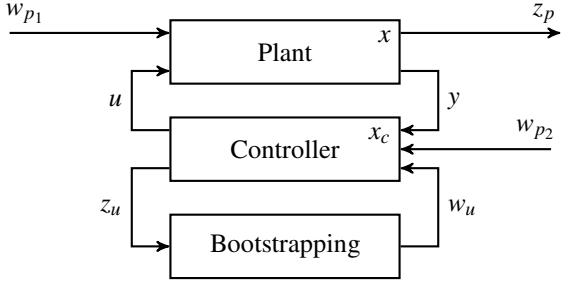

We introduce the involved system and controller before we cover the bootstrapping error within the robust control framework.

\subsection{System description}

We consider the discrete-time, linear, time-invariant system
\begin{align}
	x(t+1) &= A x(t) + Bu(t) + B_1 w_{p_1}(t)\\
	y(t) &= C x(t) + F_1 w_{p_1}(t)\\
	z_p(t) &= C_1x(t) + Eu(t) + D_1w_{p_1}(t)
\end{align}
with time-index $t\in\N$, initial condition $x(0)=x_0$, control input $u$, performance input $w_{p_1}$, measurement output $y$, and performance output $z_p$.

For this system, we consider a dynamic output feedback controller 
\begin{align}
	x_c(t+1) &= A_c x_c(t) + B_c y(t) + B_2w_{p_2}(t) + A_c w_u(t)\\
	u(t) &= C_c x_c(t) + D_c y(t) + F_2 w_{p_2}(t)\\
	z_u(t) &= x_c(t)
\end{align}
with initial state $x_c(0)=x_{c,0}$, performance input $w_{p_2}$, uncertainty input $w_u$, and uncertainty output $z_u$.
This controller emerges from a standard controller, e.g., LQG, extended by the uncertainty channel from $z_u$ to $w_u$.
This uncertainty channel will be used to incorporate the bootstrapping errors affecting the controller state into the analysis.
Further, introducing an additional performance input $w_{p_2}$ provides the opportunity to study the influence of other errors due to quantization and the cryptosystem's noise $e$. 

The interconnection of plant and controller with joint state $\xi = \begin{pmatrix}
x \\ x_c
\end{pmatrix}$ and joint performance input $w_p = \begin{pmatrix}
w_{p_1}\\ w_{p_2}
\end{pmatrix}$ results in the closed-loop system representation
\begin{align}\label{eq:clsys}
	\left(\begin{array}{c} 
		\xi(t+1) \\ \hline  z_p(t) \\ \hline z_u(t)
	\end{array}\right)
	&=
	\left(\begin{array}{c | c | c} 
		\Abc 	& \Bbc_p & \Bbc_u \\ 
		\hline  
		\Cbc_p & \Dbc_{pp} & \Dbc_{pu}\\ 
		\hline 
		\Cbc_u & \Dbc_{up} & \Dbc_{uu}
	\end{array}\right)
	\left(\begin{array}{c} 
		\xi(t)\\ \hline w_{p}(t) \\ \hline w_u(t)
	\end{array}\right)
\end{align}
with
\begin{multline}
	\left(\begin{array}{c | c | c} 
		\Abc 	& \Bbc_p & \Bbc_u \\ 
		\hline  
		\Cbc_p & \Dbc_{pp} & \Dbc_{pu}\\ 
		\hline 
		\Cbc_u & \Dbc_{up} & \Dbc_{uu}
	\end{array}\right)
	= \\
	\left(\begin{array}{cc | c c | c} 
		A + B D_c C  & B C_c 	& B_1 + BD_cF_1 & BF_2 		& 0 \\ 
		B_c C & A_c 				&  B_cF_1 & B_2 					& A_c \\ 
		\hline  
		C_1 + ED_cC & EC_c 		& D_1 + ED_cF_1 & EF_2 		& 0 \\ 
		\hline 
		0 & I 							& 0 & 0 									& 0
	\end{array}\right).
\end{multline}

\subsection{Bootstrapping error}

From a robust control viewpoint, the approximation error between the polynomial and the actual modulo function can be described by a static, time-varying uncertainty $\Delta_r$ acting on the channel from $z_u= \xi$ to $w_u = \Delta_r(z_u) = p(\xi)-(\xi \modq)$ and depending on the unknown number of overflows $r$.
By $p(\xi)$, we denote the component-wise evaluation of the polynomial.
By our choice of bootstrapping polynomial in Section~\ref{sec:BootError} with its relative error bounds, we can cover every possible bootstrapping uncertainty $\Delta_r$ as an element of the set $\boldsymbol{\Delta}$ of all uncertainties satisfying the same relative error bound.

At every time step, the input and the error of the bootstrapping satisfy the sector condition
\begin{align}\label{eq:sector}
	(w_{u,i}+\gamma z_{u,i})(\gamma z_{u,i} - w_{u,i})\geq0 
\end{align}
in every component $i$.

More generally, a sector condition of the form
\begin{align}
	\left(\Delta(z_u)-L_\ell z_u\right)^\top\left(L_u z_u - \Delta(z_u)\right) &\geq 0
\end{align}
with upper and lower bounds $L_u$ and $L_\ell$, respectively, can equivalently be described by
\begin{align}\label{eq:sectorP}
	\begin{pmatrix} 
		\Delta(z_u) \\ z_u
	\end{pmatrix}^\top
	\tau P_u
	\begin{pmatrix} 
		\Delta(z_u) \\ z_u
	\end{pmatrix} &\geq 0
\end{align}
with any $\tau>0$ and the multiplier
\begin{align}
	P_u = 
	\begin{pmatrix} 
		-2I & L_\ell + L_u \\
		L_\ell^\top + L_u^\top & -L_\ell^\top L_u - L_u^\top L_\ell
	\end{pmatrix}.
\end{align}

For the component-wise bootstrapping error, we get $L_u = \gamma I$, $L_\ell = -\gamma I$, and
\begin{align}\label{eq:P}
	P_u = 
	\begin{pmatrix} 
		-2I & 0 \\
		0 & 2 \gamma^2 I
	\end{pmatrix}.
\end{align}
Thus, the set $\boldsymbol{\Delta}$ of sector bounded uncertainties $\Delta$, containing the bootstrapping uncertainties, is described by~\eqref{eq:sectorP} with the multiplier~\eqref{eq:P}.

\subsection{Problem description}
We consider performance as in the following definition.
\begin{definition}
The system satisfies quadratic performance specified by $P_p$ if it is asymptotically stable for $w_p=0$ and there exists $\epsilon>0$ such that
\begin{align}
	\sum_{t=0}^{\infty}
	\begin{pmatrix} 
		w_p(t) \\ z_p(t)
	\end{pmatrix}^\top
	P_p
	\begin{pmatrix} 
		w_p(t) \\ z_p(t)
	\end{pmatrix} &\leq -\epsilon \sum_{t=0}^{\infty}w_p(t)^\top w_p(t)
\end{align}
for $\xi(0)=0$ and all $w_p\in\ell_2$. 
\end{definition}

This includes the $\ell_2$-gain, among other common performance specifications.

\begin{problem}
	For the closed-loop system~\eqref{eq:clsys} under the influence of bootstrapping with any possible error according to~\eqref{eq:sector}, we want to find a test for quadratic performance specified by $P_p$. 
\end{problem}

\section{Dynamic control with bootstrapping}

In this main section, we first show how we can test performance using the system and bootstrapping description from Section~\ref{sec:setup}. Then, we use a lifting approach to reduce the conservatism of our analysis.

\subsection{Stability and performance analysis}

As described in Section~\ref{sec:setup}, we consider the bootstrapping error as a time-varying, unknown, static uncertainty described by the sector condition~\eqref{eq:sectorP}.
Clearly, the zero element resembling no error, i.e., no bootstrapping, is also contained in the uncertainty description $\boldsymbol{\Delta}$.
Then, by applying robust control theory~\cite[Thm.\ 10.4.]{Scherer2000}, we can obtain the following theorem.

\begin{theorem}\label{thm:1}
	The encrypted closed-loop system~\eqref{eq:clsys} with the bootstrapping uncertainty~\eqref{eq:sectorP} satisfies robust quadratic performance with performance index $P_p = 	
	\begin{pmatrix} 
		Q_p & S_p\\
		S_p^\top & R_p
		\end{pmatrix} $ with $R_p \succeq 0$, if 
there exist $X \succ 0$ and $\tau>0$ such that
	\begin{align}
		(\star)^\top  
		\begin{pmatrix} 
			-X & 0\\
			0 & X
			\end{pmatrix} 
			\begin{pmatrix} 
			I & 0 & 0\\
			\Abc & \Bbc_p & \Bbc_u
			\end{pmatrix}  &{}\\
		+
		(\star)^\top  
		P_p
		\begin{pmatrix} 
			0 & I & 0\\
			\Cbc_p & \Dbc_{pp} & \Dbc_{pu}
			\end{pmatrix}  &{}\\
		+
		(\star)^\top  
		\tau P_u
		\begin{pmatrix} 
		0 & 0 & I\\
		\Cbc_u & \Dbc_{up} & \Dbc_{uu}
		\end{pmatrix}  &{}
		\prec 0.
	\end{align}
	\vspace*{-0.8\baselineskip}
\end{theorem}
\begin{proof}
	The proof follows directly from our uncertainty description for bootstrapping~\eqref{eq:sectorP}, which suits the robust stability and performance test in~\cite[Thm.\ 10.4.]{Scherer2000}.
\end{proof}

\subsection{Lifted dynamics}
The stability and performance test in Theorem~\ref{thm:1} is conservative for the actual encrypted system in the sense that it remains valid even if the bootstrapping error is introduced in every time step.
In reality, however, we know that we need to perform bootstrapping only every $T_{\mathrm{BS}}\in \N$ time steps. This number results from the fact that the cryptosystem supports several multiplications, i.e., control updates until the modulus has to be raised again.

Inspired by the success of lifting in~\cite{Seidel2023}, we adapt the lifting idea from~\cite{Chen1995} to obtain a system description only sampled at time instants $t=kT_{\mathrm{BS}}$ with time-index $k$ of the lifted system.

	We lift the involved signals as
\vspace*{-0.3\baselineskip}
\begin{align*}
	\tilde{z}_u(k) &= z_u(kT_{\mathrm{BS}}),  \quad
	\tilde{w}_u(k) = w_u(kT_{\mathrm{BS}}), \quad
	\tilde{\xi}(k) = \xi(kT_{\mathrm{BS}}),
	\\
	\tilde{z}_p(k) &=
	\begin{pmatrix}
		z_p(kT_{\mathrm{BS}})\\
		\vdots\\
		z_p(kT_{\mathrm{BS}}+(T_{\mathrm{BS}}-1))
	\end{pmatrix}, \\
	\tilde{w}_p(k) &=
	\begin{pmatrix}
		w_p(kT_{\mathrm{BS}})\\
		\vdots\\
		w_p(kT_{\mathrm{BS}}+(T_{\mathrm{BS}}-1))
	\end{pmatrix}.
\end{align*}

	\noindent The lifted system with time-index $k$ is represented by
		\begin{align}\label{eq:liftSys}
		\left(\begin{array}{c} 
			\tilde{\xi}(k+1) \\ \hline  \tilde{z}_p(k) \\ \hline \tilde{z}_u(k)
		\end{array}\right)
		&=
		\left(\begin{array}{c | c | c} 
			\AbcTilde
			 	& \tilde{\Bbc}_p & \tilde{\Bbc}_u \\ 
			\hline  
			\tilde{\Cbc}_p & \tilde{\Dbc}_{pp} & \tilde{\Dbc}_{pu}\\ 
			\hline 
			\Cbc_u & \tilde{\Dbc}_{up} & \Dbc_{uu}
		\end{array}\right)
		\left(\begin{array}{c} 
			\tilde{\xi}(k)\\ \hline \tilde{w}_{p}(k) \\ \hline \tilde{w}_u(k)
		\end{array}\right)
	\end{align}
	with
\bgroup \allowdisplaybreaks
\begin{align*}
	\AbcTilde &= \Abc^{T_{\mathrm{BS}}}, ~~
	\tilde{\Bbc}_u = \Abc^{T_{\mathrm{BS}}-1}\Bbc_u, ~~
	\tilde{\Bbc}_p = \begin{pmatrix}
		\Abc^{T_{\mathrm{BS}}-1}\Bbc_p &\dots & \Bbc_p
	\end{pmatrix},
	\\
	\tilde{\Cbc}_p &= \begin{pmatrix}
		\Cbc_p\\
		\vdots\\
		\Cbc_p \Abc^{T_{\mathrm{BS}}-1}
	\end{pmatrix}, \quad
	\tilde{\Dbc}_{pu} = \begin{pmatrix}
	0\\
	\Cbc_p \Bbc_{u} \\
	\vdots\\
	\Cbc_p \Abc^{T_{\mathrm{BS}}-2} \Bbc_u
	\end{pmatrix},
	\\
	\tilde{\Dbc}_{pp}  &= \begin{pmatrix}
		\Dbc_{pp} & 0 & \cdots & 0 \\
		\Cbc_p \Bbc_{p} & \multicolumn{2}{c}{\smash{\raisebox{-0.6\normalbaselineskip}{\diagdots[-25]{3.7\normalbaselineskip}{.5em}\hspace*{0.7em}}}} & \vdots \\
		\vdots & \multicolumn{1}{c}{\smash{\raisebox{-0.2\normalbaselineskip}{\diagdots[-25]{2.5\normalbaselineskip}{.5em}\hspace*{0em}}}} & \multicolumn{1}{c}{\smash{\raisebox{0.9\normalbaselineskip}{\hspace*{1.3em}\diagdots[-25]{2.5\normalbaselineskip}{.5em}}}} & 0 \\
		\Cbc_p \Abc^{T_{\mathrm{BS}}-2} \Bbc_{p} & \mathclap{\cdots} & \Cbc_p \Bbc_{p} & \Dbc_{pp}
	\end{pmatrix},\\
	\tilde{\Dbc}_{up} &= \begin{pmatrix}
		\Dbc_{up} & 0 & \cdots &0 
	\end{pmatrix}.
\end{align*}
\egroup
	
	With this equivalent system description, we can proceed with an improved stability and performance analysis.
	
	\begin{lemma}\label{lem:eq}
		The original system~\eqref{eq:clsys} satisfies quadratic performance specified by \begin{align}
			P_p = 	
		\begin{pmatrix} 
			Q_p & S_p\\
			S_p^\top & R_p
		\end{pmatrix},
	\end{align}
		if and only if the lifted system~\eqref{eq:liftSys} satisfies quadratic performance specified by \begin{align}
			\tilde{P}_p = 
			\begin{pmatrix} 
				I_{T_{\mathrm{BS}}}\otimes Q_p & I_{T_{\mathrm{BS}}}\otimes S_p\\
				I_{T_{\mathrm{BS}}}\otimes S_p^\top & I_{T_{\mathrm{BS}}}\otimes R_p
				\end{pmatrix} .
		\end{align}
		\vspace*{0\baselineskip}
	\end{lemma}
	\begin{proof}
		Similar results were derived in~\cite{Linsenmayer2017} (stability), \cite{Seidel2023} ($\ell_2$-performance), and~\cite{Lang2024} (the general case). 
		The proof follows from
		\begin{align}
			\sum_{t=0}^{\infty}
			\begin{pmatrix} 
				w_p(t) \\ z_p(t)
			\end{pmatrix}^\top
			P_p
			\begin{pmatrix} 
				w_p(t) \\ z_p(t)
			\end{pmatrix}
			&\,=\,
			\sum_{k=0}^{\infty}
			\begin{pmatrix} 
				\tilde{w}_p(k) \\ \tilde{z}_p(k)
			\end{pmatrix}^\top
			\tilde{P}_p
			\begin{pmatrix} 
				\tilde{w}_p(k) \\ \tilde{z}_p(k)
			\end{pmatrix}
		\end{align}
		and 
		\begin{align}
			\sum_{t=0}^{\infty}w_p(t)^\top w_p(t)
			=
			\sum_{k=0}^{\infty}\tilde{w}_p(k)^\top \tilde{w}_p(k).
		\end{align}
	\end{proof}
	
	This performance equivalence can be used for our second theorem.
	
	\begin{theorem}\label{thm:2}
		The encrypted closed-loop system~\eqref{eq:clsys} with the bootstrapping uncertainty~\eqref{eq:sectorP} satisfies robust quadratic performance with performance index $P_p = 	
		\begin{pmatrix} 
			Q_p & S_p\\
			S_p^\top & R_p
		\end{pmatrix} $ with $R_p \succeq 0$, if 
		there exist $X \succ 0$ and $\tau>0$  such that
		\begin{align}
			(\star)^\top  
			\begin{pmatrix} 
				-X & 0\\
				0 & X
				\end{pmatrix} 
				\begin{pmatrix} 
				I & 0 & 0\\
				\AbcTilde & \tilde{\Bbc}_p & \tilde{\Bbc}_u
				\end{pmatrix}  &{}\\
			+
			(\star)^\top  
			\tilde{P}_p
			\begin{pmatrix} 
				0 & I & 0\\
				\tilde{\Cbc}_p & \tilde{\Dbc}_{pp} & \tilde{\Dbc}_{pu}
				\end{pmatrix}  &{}\\
			+
			(\star)^\top  
			\tau P_u
			\begin{pmatrix} 
				0 & 0 & I\\
				\Cbc_u & \tilde{\Dbc}_{up} & \Dbc_{uu}
				\end{pmatrix}  &{}
			\prec 0
		\end{align}
		with
		\begin{equation}
			\tilde{P}_p = 
			\begin{pmatrix} 
				I_{T_{\mathrm{BS}}}\otimes Q_p & I_{T_{\mathrm{BS}}}\otimes S_p\\
				I_{T_{\mathrm{BS}}}\otimes S_p^\top & I_{T_{\mathrm{BS}}}\otimes R_p
			\end{pmatrix} .
		\end{equation}
		\vspace*{-0.8\baselineskip}
	\end{theorem}
	\begin{proof}
		The proof follows from Theorem~\ref{thm:1} and Lemma~\ref{lem:eq}.
	\end{proof}

Theorem~\ref{thm:2} yields a less conservative test than Theorem~\ref{thm:1} since its lifted system description~\eqref{eq:liftSys} captures the fact that the bootstrapping errors 
only occur
every $T_{\mathrm{BS}}$ time steps.

\subsection{Numerical example}
For the numerical evaluation we use the system 
with

\vspace*{-0.8\baselineskip}
\begin{alignat*}{4}
	A &= \begin{pmatrix}
		-0.5 & 0.1\\
		0& -0.2
	\end{pmatrix}, &~~
	B &= \begin{pmatrix}
		0\\
		1
	\end{pmatrix}, &~~
	B_1 &= \begin{pmatrix}
		1\\
		1
	\end{pmatrix}, &~~
	E &= \begin{pmatrix}
		1\\
		1
	\end{pmatrix},\\
	C &= \begin{pmatrix}
		1 & 0
	\end{pmatrix}, &~~
	F_1 &= 0, &~~
	C_1 &= I, &~
	D_1 &= 0,\\
	A_c &= \begin{pmatrix}
		0.13 & 0.1\\
		-1.27& 0.15
	\end{pmatrix}, &~~
	B_c &= \begin{pmatrix}
		0.63\\
		-0.27
	\end{pmatrix}, &~~
	B_2 &= I,\\
	C_c &= \begin{pmatrix}
		-1 & 0.35
	\end{pmatrix}, &~~
	D_c &= 0, &
	F_2 &=0.
\end{alignat*}
\noindent Using the polynomial from Section~\ref{sec:BootError} with $T_{\mathrm{BS}}=10$, an upper bound on the $\ell_2$-gain is found with $Q_p = -\gamma_{\ell_2}^2I,  S_p = 0$, and $R_p = I$. Theorem~\ref{thm:1} yields $\gamma_{\ell_2}=5.13$, and the less conservative Theorem~\ref{thm:2} returns $\gamma_{\ell_2}=3.97$. A simulation over 10.000 time steps yields an empirical lower bound of $\gamma_{\ell_2}=1.88$.

\section{Analysis of Reset and FIR controllers}
	It is interesting to note that if the sector slope is chosen as $\gamma = 1$ in the uncertainty description~\eqref{eq:sector}, then also reset controllers as in~\cite{Murguia2020} can be analyzed using Theorem~\ref{thm:2}. This is because this sector includes minus identity as part of its error description, and this is precisely the error introduced by resetting the entire controller state to zero.
	To pursue the analysis in this case, we can use a nominal controller with its matrices $A_c, B_c, C_c$, and $D_c$, and treat the reset purely by the uncertainty channel.

	Similarly, after small modifications, also FIR controllers of length $N$ such as in~\cite{Schluter2021} can be analyzed by Theorem~\ref{thm:1}. 
	Instead of resetting all states, the dropped measurement $y(t-N)$ is fed into the uncertainty channel as $z_u(t)$ at each time step. The uncertainty input yields $w_u(t) = -z_u(t)$, which can correspond to $\gamma=1$ in the sector notation. The last necessary modification is the change of $\Bbc_u=\begin{pmatrix}
		0\\ A_c^N B_c
	\end{pmatrix}$. 
	This cancels the effect that $y(t-N)$ would have had on the current controller state.
	Thus, if nominal dynamic controllers are used in a reset or FIR fashion, our theorems provide stability and performance tests.

\section{Summary and Outlook}

In this paper, we have shown how bootstrapping can be incorporated into the stability and performance analysis of encrypted control systems. For the bootstrapping polynomial in the core of the bootstrapping operation, we derived an appropriate error description. Using robust control theory, we were able to derive a stability and performance test.
By finding a lifted system description, which considers the error only at times when bootstrapping is done, a second, less conservative test was introduced.

This is the first paper to explicitly incorporate the bootstrapping effects into system analysis. This joint analysis offers several new possibilities.
If the controller is robust enough to allow for large bootstrapping errors, a less precise bootstrapping polynomial is required. In this case, the degree of the bootstrapping polynomial can be reduced, offering less precision but a computationally less complex encrypted evaluation. This can lead to quicker bootstrapping, which has been the main restriction of applying bootstrapping in practice, so far. 
Moreover, for better performance of the control system, it can be beneficial to extend the time between two bootstrapping operations. This would lead to a less frequent introduction of bootstrapping errors.
 Whether this is possible in practice has to be investigated from a cryptographical point of view.

\bibliographystyle{IEEEtran} %
\bibliography{BibForCDC24} %

\end{document}